\documentclass[letterpaper, 10 pt, conference]{ieeeconf}  

\IEEEoverridecommandlockouts                              

\usepackage[T1]{fontenc}

\usepackage{amsmath,amssymb,amsfonts}
\usepackage{mathtools}

\usepackage{cases}
\usepackage{graphicx}
\usepackage{float}
\usepackage{color}
\usepackage{textcomp}

\usepackage{algorithm}
\usepackage{algpseudocode}

\usepackage{cite}

\newtheorem{theorem}{Theorem}

\newtheorem{definition}{Definition}
\newtheorem{lemma}{Lemma}
\newtheorem{remark}{Remark}

\title{\LARGE \bf
On the Reduced Rank Hypersurface Germ Condition for Near-Controllability of Bilinear Systems
}

\author{Wenyu Zhao$^{1}$ and Lin Tie$^{1}$
\thanks{*This work was supported by the National Natural Science Foundation of China (62373023).}
\thanks{$^{1}$The authors are with the School of Automation Science and Electrical Engineering, Beihang University, 100083, Beijing, P. R. China 
        {\tt\small WZhao@buaa.edu.cn, tielin@buaa.edu.cn}}%
}

\begin{document}

\maketitle
\thispagestyle{empty}
\pagestyle{empty}

\begin{abstract}
Controllability of nonlinear systems has been extensively studied by using the Lie algebra methods. Although local controllability can be proved, global controllability is in general hard to obtain even for bilinear systems. Recently, a new approach is developed to study controllability of bilinear systems by checking whether the reduced rank points corresponding to the bilinear terms form hyperplanes or hypersurfaces, so that near-controllability, in the global sense, can be algebraically verified once no hypersurface exists. In this paper, we show that there is a flaw in the algorithm for checking the reduced rank hypersurface germ condition (RRHGC) that some special cases are overlooked. We thus propose a new algorithm for testing the RRHGC based on a Darboux-polynomial characterization, which can be used to deal with all cases. An example is provided to demonstrate the proposed new algorithm.
\end{abstract}

$\left. {}\right. $
\section{Introduction}

Since the concept of controllability was introduced in the 1960s \cite{kalman1960general,kalman1963mathematical}, the controllability theory of linear systems has been well established and widely applied in control engineering \cite{kailath1980linear, wonham1985linear, klamka2008controllability} and complex networks \cite{liu2011controllability,ruths2014control}. In contrast, the controllability of nonlinear systems remains difficult to characterize because of nonlinear dynamics, and many controllability problems for nonlinear and even bilinear systems remain unresolved. With the theory of differential geometry, substantial progress was made in the 1970s and thereafter \cite{elliott1971consequence,sussmann1972controllability,brockett1972system,bruni1974bilinear,boothby1975transitivity,hirschorn1975controllability,hermann1977nonlinear,isidori1995nonlinear,jurdjevic1997geometric,elliott2009bilinear}. These studies provide controllability criteria through a geometric analysis of the Lie algebras generated by vector fields. Such methods establish controllability when the vector fields are symmetric. For the more general nonsymmetric cases, however, they typically guarantee only accessibility, which is weaker than controllability.

The nonsymmetrical cases are also intractable to deal with even for bilinear systems. A central difficulty in studying controllability of bilinear systems is how to extend local controllability to a global result without additional conditions making the vector fields symmetric. To address it, we proposed a completely new approach. Specifically, we first solved the near-controllability problem of strictly bilinear systems through the concept of near-transitivity, and analyzed the influence of reduced rank points on the connectedness of the state space in \cite{zhao2024near}. Then, for homogeneous bilinear systems, the connectedness is recovered by using the drift term in \cite{zhao2026near}. As a result, near-controllability can be verified in the global sense as long as the reduced rank points form hyperplanes instead of hypersurfaces.

In fact, the significance of the new approach proposed in \cite{zhao2026near} lies in its departure from the conventional strategy of imposing additional conditions to convert a transition semigroup into a group and then establishing controllability through Lie algebraic methods. We showed that even when the transition matrices form only a semigroup, controllability problems can still be addressed globally within the framework of near-controllability. One of the technical steps in this approach is to determine whether the reduced rank points of the corresponding strictly bilinear system consist only of hyperplanes without any hypersurface. This does not mean that near-controllability fails in the cases of hypersurfaces, but the proofs in these cases will be rather complicated. Therefore, \cite{zhao2026near} classified the eigenvalues of the systems’ coefficient matrices and proposed an algorithm for checking whether there exists a hypersurface in the set formed by the reduced rank points. This algorithm works for most cases but overlooks certain special cases. To address this issue, we propose a new algorithm to include the overlooked cases in this paper. That is, for matrices with only real eigenvalues, the new algorithm unifies Propositions 3.1 and 3.2 of \cite{zhao2026near} and provides a more efficient verification procedure. For matrices with complex eigenvalues, the relevant invariant hypersurfaces are characterized through Darboux polynomials and the RRHGC verification is formulated as the problem of finding a Darboux polynomial with a constant cofactor in the new algorithm.

The rest of this paper is organized as follows. Section II analyzes the conditions for satisfying the RRHGC based on the spectral classification of the coefficient matrix. When the coefficient matrix $B$ has an entirely real spectrum, it satisfies the RRHGC if and only if at least one of the five proposed conditions holds. When $B$ has complex eigenvalues, RRHGC verification is reformulated as the problem of finding a feasible solution to some equations and inequalities. Based on the results of Section II, Section III presents a verification algorithm and provides an illustrative example. Finally, Section IV concludes the paper.

\section{Analysis of the RRHGC via Spectrum Classification}

\subsection{Problem Formulation}
Let
\begin{equation*}
    X_B(x)=Bx,
\end{equation*}
where $B\in\mathbb R^{n\times n}$ and $x\in\mathbb R^n$. Let
\begin{equation}
    \mathcal L_B q(x):=\nabla q(x)^TBx \label{def lie derive}
\end{equation}
denote the Lie derivative of a polynomial $q$ along $X_B$.

According to Definition 3.2 in \cite{zhao2026near}, $\mathcal{S}$ is a reduced rank hypersurface germ for $X_B$ if: 1) $X_B(x)\in T_x\mathcal{S}$ for every $x\in\mathcal{S}$ and 2) The equation of $\mathcal{S}$ is homogeneous and its degree $d$ satisfies $2\le d\le n$. This definition can be reformulated as follows.

\begin{definition}[Reduced Rank Hypersurface Germ]
\label{def:admissible-RRHG}
A reduced rank hypersurface germ of $X_B$ is an equation
\begin{equation*}
    q(x)=0
\end{equation*}
where the following three conditions are satisfied:
\begin{enumerate}
    \item $q\in\mathbb R[x_1,\ldots,x_n]$ is irreducible and
          homogeneous of degree $d$, with $2\le d\le n$;
    \item There exists $x^*\in\mathbb R^n\setminus\{0\}$ such that
          \begin{equation*}
              q(x^*)=0,\quad \nabla q(x^*)\ne0;
          \end{equation*}
    \item $X_B(x)$ is tangent to $q=0$ at every regular point of the
          hypersurface.
\end{enumerate}
The matrix $B$ is said to satisfy the RRHGC if $X_B$ admits such a
hypersurface germ.
\end{definition}

In Definition 1, irreducibility excludes a union of lower-degree algebraic components. The lower degree bound excludes hyperplanes, and the upper bound agrees with the degree restriction in the definition used in \cite{zhao2026near}. The regular-zero condition ensures that the real zero set contains a smooth germ of codimension one.

\begin{lemma}
\label{lem:invariance-identity}
Let $q\in\mathbb R[x_1,\ldots,x_n]$ be irreducible and homogeneous, and
suppose that $q=0$ has a real regular point. Then $q=0$ is invariant
under $X_B$ if and only if there exists a constant $\kappa\in\mathbb R$
such that
\begin{equation}
    \mathcal L_Bq=\kappa q. \label{darboux}
\end{equation}
\end{lemma}

\begin{proof}
(Necessity) Since $X_B(x)$ is tangent to $q=0$ at every regular point of $q(x)$,
\begin{equation*}
    \mathcal L_Bq(x)=0
\end{equation*}
whenever $q(x)=0$ and $\nabla q(x)\ne0$. Since $q$ is irreducible and has a nonsingular real zero, the regular real locus of $q=0$ is Zariski dense in the irreducible algebraic hypersurface defined by $q$. Thus, $\mathcal L_Bq$ vanishes on a Zariski dense subset of that hypersurface, and
\begin{equation*}
    q\mid \mathcal L_Bq.
\end{equation*}
The vector field is linear, so $\mathcal L_Bq$ is homogeneous of the
same degree as $q$. The quotient is consequently a real constant,
which proves \eqref{darboux}.

(Sufficiency) Suppose \eqref{darboux} hold and $x(t)$ is a
solution of $\dot x=Bx$. Then
\begin{equation*}
    \frac{d}{dt}q(x(t))=\kappa q(x(t)),
\end{equation*}
which yields
\begin{equation}
    q(e^{tB}x_0)=e^{\kappa t}q(x_0).
    \label{eq:Darboux-flow}
\end{equation}
It follows that every trajectory starting on $q=0$ remains on $q=0$.
\end{proof}

\begin{remark}
Equation \eqref{darboux} can be interpreted in terms of Darboux polynomials. For a polynomial vector field $X$, a non-constant polynomial $q$ is called a Darboux polynomial if
\begin{equation*}
    X(q)=Kq
\end{equation*}
for some polynomial $K$, called its cofactor \cite{darboux1878memoire, Zhang2017Integrability}. Therefore, \eqref{darboux} states precisely that $q$ is a Darboux polynomial of the linear vector field $X_B(x)=Bx$ with the constant cofactor $K=\kappa$.
\end{remark}

It should be noted that the RRHGC is invariant under real similarity transformations. Actually, if $B=PJP^{-1}$ and $x=Py$, then $q(y)=0$ is invariant under $\dot y=Jy$ if and only if $q(P^{-1}x)=0$ is invariant under $\dot x=Bx$. Invertible linear transformations preserve degree, irreducibility, and the existence of a real regular zero. Therefore, we consider the canonical form in the following text.

\subsection{Real-Spectrum Case}
\label{subsec:real-spectrum}

Suppose that all eigenvalues of $B$ are real, whose canonical form is given by
\begin{equation}
    J=\bigoplus_{j=1}^{s}\bigoplus_{\ell=1}^{g_j}
       J_{m_{j\ell}}(\lambda_j),
    \label{eq:real-Jordan-form}
\end{equation}
where $\lambda_1,\ldots,\lambda_s$ are pairwise distinct, $g_j$ is the
geometric multiplicity of $\lambda_j$, and $m_{j\ell}$ is the dimension of the
$\ell$-th block corresponding to $\lambda_j$. Let $h$ denote the number of Jordan blocks with dimension greater than one. For $c=(c_1,\ldots,c_s)\in\mathbb Z^s$, define
\begin{equation*}
    c_j^+:=\max\{c_j,0\},\quad
    c_j^-:=\max\{-c_j,0\}.
\end{equation*}

\begin{theorem}[Real-Spectrum Case]
\label{thm:five-case-classification}
Suppose that $B$ has only real eigenvalues. Then $B$ satisfies the RRHGC if
and only if at least one of the following five conditions
\textnormal{(R1)}, \textnormal{(R2)}, \textnormal{(R3)},
\textnormal{(R4)}, and \textnormal{(R5)} holds:
\begin{align*}
\mathrm{(R1)}\quad &max_{j,\ell}m_{j\ell}\ge3;\\
\mathrm{(R2)}\quad &h\ge2;\\
\mathrm{(R3)}\quad &g_j\ge3\quad\text{for some }j;\\
\mathrm{(R4)}\quad &g_j\ge2\ \text{and}\ g_k\ge2
                     \quad\text{for some }j\ne k;\\
\mathrm{(R5)}\quad
&\begin{aligned}[t]
  &\text{there exists }c\in\mathbb Z^s\setminus\{0\}
    \text{ such that}\\
  &\sum_{j=1}^{s}c_j=0,
    \quad \sum_{j=1}^{s}c_j\lambda_j=0, \\
    &
    2\le \sum_{j=1}^{s}c_j^+
       =\sum_{j=1}^{s}c_j^-\le n,\\
    & \gcd(|c_1|,\ldots,|c_s|)=1.
  \end{aligned}
\end{align*}
\end{theorem}

\begin{proof}
    (Sufficiency) \paragraph*{Case \textnormal{(R1)}}
Suppose that there exists a Jordan block with dimension greater than two. Since only the last three coordinates are used, suppose
\begin{equation*}
    X(x)=Jx=\begin{bmatrix}
        \lambda&1&0\\
        0&\lambda&1\\
        0&0&\lambda\\
    \end{bmatrix}
\end{equation*}
without loss of generality, where $x=\left[x_1\quad x_2\quad x_3\right]^T$. Their dynamics are given by
\begin{equation*}
    \dot x_1=\lambda x_1+x_2,\quad
    \dot x_2=\lambda x_2+x_3,\quad
    \dot x_3=\lambda x_3.
\end{equation*}
Define
\begin{equation}
    q=x_2^2-2x_1x_3.
    \label{germ R1}
\end{equation}
Direct calculation gives
\begin{equation*}
    \mathcal L_Jq=2\lambda q.
\end{equation*}
The quadratic form in \eqref{germ R1} is irreducible over $\mathbb R$. At $x=\left[0\quad 0\quad 1\right]^T$, it satisfies $q=0$ and $\nabla q^T=\left[-2\quad 0\quad 0\right]\ne\textbf{0}$. Thus, equation \eqref{germ R1} is an admissible reduced rank hypersurface germ for $X_J$.

\paragraph*{Case \textnormal{(R2)}}
Choose the last two coordinates from two nontrivial Jordan blocks, given by
\begin{equation*}
\begin{aligned}
    \dot z_1&=\lambda z_1+z_2,&
    \dot z_2&=\lambda z_2,\\
    \dot w_1&=\mu w_1+w_2,&
    \dot w_2&=\mu w_2.
\end{aligned}
\end{equation*}
Let
\begin{equation}
    q=z_1w_2-z_2w_1.
    \label{germ R2}
\end{equation}
Then
\begin{equation*}
    \mathcal L_Jq=(\lambda+\mu)q
\end{equation*}
with
\begin{equation*}
    J=\begin{bmatrix}
        \lambda&1&0&0\\
        0&\lambda&0&0\\
        0&0&\mu&1\\
        0&0&0&\mu\\
    \end{bmatrix}.
\end{equation*}
The polynomial in \eqref{germ R2} is irreducible. At the point
\begin{equation*}
    (z_1,z_2,w_1,w_2)=(1,0,0,0),
\end{equation*}
the value of $q$ is zero and its gradient is nonzero. Consequently, \eqref{germ R2} defines an admissible germ.

\paragraph*{Case \textnormal{(R3)}}
Choose three coordinates $z_1,z_2,z_3$ as the last coordinate of three Jordan blocks corresponding to the same eigenvalue $\lambda_j$. The dynamics of each $z_i$ satisfy $\dot z_i=\lambda_jz_i$. The polynomial
\begin{equation}
    q=z_1^2+z_2^2-z_3^2
    \label{germ R3}
\end{equation}
satisfies $\mathcal L_Jq=2\lambda_jq$, where
\begin{equation*}
    J=\begin{bmatrix}
        \lambda_j&0&0\\
        0&\lambda_j&0\\
        0&0&\lambda_j\\
    \end{bmatrix}.
\end{equation*}
It can be verified that \eqref{germ R3} is an admissible germ.

\paragraph*{Case \textnormal{(R4)}}
Choose terminal coordinates $z_1,z_2$ associated with $\lambda_j$ and
$w_1,w_2$ associated with $\lambda_k$, where $j\ne k$. Set
\begin{equation}
    q=z_1w_1+z_2w_2.
    \label{germ R4}
\end{equation}
Then
\begin{equation*}
    \mathcal L_Jq=(\lambda_j+\lambda_k)q,
\end{equation*}
where
\begin{equation*}
    J=\begin{bmatrix}
        \lambda_j&0&0&0\\
        0&\lambda_j&0&0\\
        0&0&\lambda_k&0\\
        0&0&0&\lambda_k\\
    \end{bmatrix}.
\end{equation*}
\eqref{germ R4} can also be proved to be admissible.

\paragraph*{Case \textnormal{(R5)}}
For each $\lambda_j$, select the last coordinate $z_j$ from a Jordan
block of $\lambda_j$. Let
\begin{equation}
    q=m_1\prod_{j=1}^{s}z_j^{c_j^+}
           +m_2\prod_{j=1}^{s}z_j^{c_j^-},
    \label{germ R5}
\end{equation}
where $m_1$ and $m_2$ are nonzero real coefficients. Both monomials in \eqref{germ R5} have the same degree. Since
\begin{equation*}
    c_j=c_j^+-c_j^-,
\end{equation*}
it follows from Condition (R5) that
\begin{equation*}
    \sum_jc_j^+\lambda_j=\sum_jc_j^-\lambda_j.
\end{equation*}
Note that
\begin{equation*}
    \mathcal{L}_Jz_j^{c_j^+}=c_j^+z_j^{c_j^+-1}\lambda_jz_j=\lambda_jc_j^+z_j^{c_j^+}.
\end{equation*}
Therefore,
\begin{align*}
    \mathcal L_Jq&=\mathcal L_J\left(m_1\prod_{j=1}^{s}z_j^{c_j^+}\right)+\mathcal L_J\left(m_2\prod_{j=1}^{s}z_j^{c_j^-}\right)\\
    &=m_1\sum_jc_j^+\lambda_j\prod_kz_k^{c_k^+}+m_2\sum_jc_j^-\lambda_j\prod_kz_k^{c_k^-}\\
    &=\left(\sum_jc_j^+\lambda_j\right)q
\end{align*}
where $J=\text{diag}(\lambda_1,\ldots,\lambda_s)$. We now prove that $q$ is irreducible. The two monomials in $q$ have disjoint supports and hence have no non-constant monomial common factor. Let
\begin{align*}
    P&:=\mathbb R[z_1,\ldots,z_s],\\
    L&:=\mathbb R[z_1^{\pm1},\ldots,z_s^{\pm1}]
\end{align*}
denote the ordinary polynomial ring and the corresponding Laurent polynomial ring, respectively. Let $z^{c^+}$ denote $\prod_{j=1}^{s}z_j^{c_j^+}$ and $z^{c^-}$ denote $\prod_{j=1}^{s}z_j^{c_j^-}$. Since $c=c^+-c^-$, we have
\begin{equation*}
    q=m_1z^{c^+}+m_2z^{c^-}=z^{c^-}(m_1z^c+m_2)\in L.
\end{equation*}
The Laurent monomial $z^{c^-}$ is a unit in $L$, with inverse $z^{-c^-}$. Hence it suffices to prove that $m_1z^c+m_2$ is irreducible in $L$.

Since
\begin{equation*}
    \gcd(|c_1|,\ldots,|c_s|)=1,
\end{equation*}
there exists a unimodular matrix
\begin{equation*}
    U\in\operatorname{SL}(s, \mathbb Z)
\end{equation*}
such that
\begin{equation*}
    Uc=e_1,
\end{equation*}
where $e_1=[1\quad 0\quad\cdots\quad 0]^T$. The matrix $U$ induces an automorphism of Laurent polynomial rings by
\begin{equation*}
    z^\alpha\longmapsto w^{U\alpha},
    \quad \alpha\in\mathbb Z^s.
\end{equation*}
This change of variables is invertible because
$U^{-1}\in\operatorname{SL}(s, \mathbb Z)$. Under this automorphism,
\begin{equation*}
    m_1z^c+m_2\longmapsto m_1w^{Uc}+m_2=m_1w_1+m_2.
\end{equation*}
Moreover,
\begin{equation*}
    \frac{
      \mathbb R[w_1^{\pm1},\ldots,w_s^{\pm1}]
    }{(m_1w_1+m_2)}
    \cong
    \mathbb R[w_2^{\pm1},\ldots,w_s^{\pm1}],
\end{equation*}
and the ring on the right-hand side is an integral domain. Therefore $(m_1w_1+m_2)$ is a prime ideal, so $m_1w_1+m_2$ is irreducible. Since the above Laurent change of variables is an automorphism, $m_1z^c+m_2$ is irreducible in $L$. Consequently, $q$ is also irreducible in $L$ up to multiplication by a Laurent unit.

We finally transfer this conclusion back to the ordinary polynomial
ring $P$. Suppose, to the contrary, that \eqref{germ R5} admits a
nontrivial factorization
\begin{equation*}
    q=f g,
    \quad f,g\in P.
\end{equation*}
Since $P\subset L$, this is also a factorization in $L$. The irreducibility of $q$ in $L$ implies that one of the two factors, say $f$, must be a unit in $L$. The units of $L$ are precisely the Laurent monomials $a z^\nu$, where $a\in\mathbb R\setminus\{0\}$, $\nu\in\mathbb Z^s$. Since $f$ belongs to the ordinary polynomial ring $P$, all components of $\nu$ must be nonnegative. Thus $f$ is a nonzero constant times an ordinary monomial.

Since this monomial divides
\begin{equation*}
    q=m_1z^{c^+}+m_2z^{c^-}
\end{equation*}
in $P$, it must divide both $z^{c^+}$ and $z^{c^-}$. However, for each $j$, the definitions of $c_j^+$ and $c_j^-$ give
\begin{equation*}
    \min\{c_j^+,c_j^-\}=0.
\end{equation*}
It follows that
\begin{equation*}
    \gcd(z^{c^+},z^{c^-})=1.
\end{equation*}
Hence the only ordinary monomial that divides both terms is a constant. This contradicts the assumption that $f$ is non-constant. Therefore, $q$ is irreducible in $\mathbb R[z_1,\ldots,z_s]$.

We now verify Condition 2) in Definition 1. Since $\gcd(|c_1|,\ldots,|c_s|)=1$, there exists an index $k$ such that $c_k$ is odd. Setting $z_j^*=1$ for $j\neq k$ and choosing $z_k^*\in\mathbb{R}\setminus\{0\}$ such that $(z_k^*)^{c_k}=-m_2/m_1$, we obtain $q(z^*)=0$. Moreover,
\[
\frac{\partial q}{\partial z_k}(z^*)
=-m_2(z^*)^{c^-}\frac{c_k}{z_k^*}\neq \textbf{0},
\]
so $z^*$ is a real regular zero of $q$. Finally, note that the degree of $q$
\begin{equation*}
    d(q)=\sum_{j=1}^sc_j^+=\sum_{j=1}^sc_j^-
\end{equation*}
satisfies $2\le d(q)\le n$, then we can conclude that $q$ is an admissible reduced rank hypersurface germ.

The sufficiency proof of Theorem 1 has been completed.

(Necessity) Suppose that $B$ satisfies the RRHGC. Consider the canonical form shown in \eqref{eq:real-Jordan-form}. Let $q$ be an admissible irreducible
homogeneous polynomial of degree $d$. By
Lemma~\ref{lem:invariance-identity},
\begin{equation}
    \mathcal L_Jq=\kappa q
    \label{operator L_J}
\end{equation}
for some $\kappa\in\mathbb R$.

Rewrite $J$ as the sum of a diagonal matrix and an upper triangular matrix, given by
\begin{equation*}
    J=S+N,
\end{equation*}
where $S$ is diagonal, $N$ is nilpotent, and $SN=NS$. Let $\mathcal H_d$ denote the finite-dimensional space of degree-$d$ homogeneous polynomials, given by
\begin{equation*}
    \mathcal H_d
    :=\{q\in\mathbb R[x_1,\ldots,x_n]\mid q
       \text{ is homogeneous of degree }d\}.
\end{equation*}

The induced operators of $\mathcal L_S$ and $\mathcal L_N$ on $\mathcal H_d$ satisfy
\begin{equation}
    \mathcal L_J=\mathcal L_S+\mathcal L_N, \label{operator decom}
\end{equation}
where $[\mathcal L_S,\mathcal L_N]=0$. The operator $\mathcal L_S$ is semisimple in the monomial basis. The operator $\mathcal L_N$ is nilpotent because it is the infinitesimal action induced by the nilpotent linear map $N$ on $\mathcal H_d$. Thus, they are respectively the semisimple and nilpotent parts of $\mathcal L_J$.

Since $\mathcal L_S$ is semisimple, $\mathcal H_d$ can be expressed as the direct sum of eigenspaces of $\mathcal L_S$, i.e.,
\begin{equation*}
    \mathcal H_d=\bigoplus_\eta E_\eta,\quad E_\eta=\ker (\mathcal L_S-\eta I).
\end{equation*}
Thus, $q$ can be decomposed as
\begin{equation}
    q=\sum_\eta q_\eta, \label{decom q}
\end{equation}
where $q_\eta\in E_\eta$ and
\begin{equation}
    \mathcal L_S q_\eta=\eta q_\eta. \label{q_eta eigen}
\end{equation}
By \eqref{operator decom}, for any $q_\eta\in E_\eta$,
\begin{equation*}
    \mathcal{L}_S(\mathcal{L}_Nq_\eta)=\mathcal{L}_N(\mathcal{L}_Sq_\eta)=\mathcal{L}_N(\eta q_\eta)=\eta\mathcal{L}_Nq_\eta.
\end{equation*}
Therefore, $\mathcal{L}_Nq_\eta$ is also in the eigenspace $E_\eta$. Combining \eqref{operator L_J}, \eqref{operator decom} and \eqref{decom q} yields
\begin{equation*}
    (\mathcal{L}_S+\mathcal{L}_N)\sum_\eta q_\eta=\kappa\sum_\eta q_\eta.
\end{equation*}
It follows from \eqref{q_eta eigen} that
\begin{equation*}
    \sum_\eta((\eta-\kappa)I+\mathcal{L}_N)q_\eta=0.
\end{equation*}
Since $q$ is the direct sum of $q_\eta$, we obtain
\begin{equation}
    ((\eta-\kappa)I+\mathcal{L}_N)q_\eta=0 \label{eta-kappa}
\end{equation}
for each $\eta$.

On the eigenspace with eigenvalue $\eta\ne\kappa$, since $\mathcal{L}_N$ is nilpotent, whose eigenvalues are all zero, which implies that $q_\eta=0$. Therefore, the projection of $q$ onto $E_\eta$ is zero when $\eta\ne\kappa$. That is, the only nonzero component is in the eigenspace $E_\kappa$, which means that
\begin{equation*}
    \mathcal{L}_Sq=\kappa q.
\end{equation*}

Then, on the eigenspace with eigenvalue $\eta=\kappa$, it follows from \eqref{eta-kappa} that $\mathcal{L}_Nq_\kappa=0$. Since $q_\eta=0$ when $\eta\ne\kappa$, we have
\begin{equation}
    \mathcal{L}_Nq=0. \label{operator L_N}
\end{equation}

We now prove that at least one of the five conditions holds. Suppose that none of the five conditions holds. The failure of \textnormal{(R1)} and \textnormal{(R2)} implies
that all Jordan blocks have size at most two and that at most one block
has size two. Hence either $N=0$ or
\begin{equation*}
    \mathcal L_N=y_2\frac{\partial}{\partial y_1},
\end{equation*}
where $y_1, y_2$ denote the coordinates of the unique nontrivial block. By \eqref{operator L_N}, we obtain
\begin{equation*}
    \frac{\partial q}{\partial y_1}=0.
\end{equation*}
Thus, $q$ depends only on terminal Jordan-chain coordinates. Denote the $g_j$ terminal coordinates
associated with $\lambda_j$ by
\begin{equation*}
    z_{j,1},\ldots,z_{j,g_j}.
\end{equation*}

For a monomial $M$ with respect to these variables, let
\begin{equation*}
    \beta_j(M):=\deg_{z_{j,1},\ldots,z_{j,g_j}}M
\end{equation*}
be its total degree in the $j$th eigenvalue group. Since $\mathcal L_Sq=\kappa q$, every monomial in $q$ satisfies
\begin{equation}
\begin{aligned}
    \sum_{j=1}^{s}\beta_j(M)=d,\quad
    \sum_{j=1}^{s}\beta_j(M)\lambda_j=\kappa.
    \label{degree-weight}
\end{aligned}
\end{equation}

Suppose that two monomials $M$ and $M'$ in $q$ corresponding to different vectors $\beta$ and $\beta'$. Let $c=\beta-\beta'$ and divide $c$ by the greatest common divisor of its
nonzero entries. We derive from \eqref{degree-weight} that
\begin{equation*}
    \sum_jc_j=0,
    \quad
    \sum_jc_j\lambda_j=0.
\end{equation*}
Moreover,
\begin{equation*}
    d(c)=\sum_jc_j^+=\sum_jc_j^-\le d\le n.
\end{equation*}
The case $d(c)=1$ would imply $\lambda_j=\lambda_k$ for two distinct indices, contrary to the definition of the distinct eigenvalues $\lambda_1,\ldots,\lambda_s$. Therefore $d(c)\ge2$, and condition \textnormal{(R5)} holds, which yields a contradiction.

It follows that, when \textnormal{(R5)} fails, all monomials in $q$ have the same vector $\beta$. Thus, $q$ is multihomogeneous with a fixed degree in every eigenvalue group. The failure of \textnormal{(R3)} and \textnormal{(R4)} implies that $g_j\le2$ for every $j$ and at most one eigenvalue has a geometric multiplicity of 2. Since $q$ is multihomogeneous, every one-variable group contributes a common monomial factor to $q$. Note that $q$ is irreducible and has degree at least two, then all such group degrees must be zero. Consequently, $q$ can depend only on the two variables of a single group, if such a group exists.

An irreducible real homogeneous polynomial in two variables has degree one or two. In the degree-two case it is a definite quadratic form. Otherwise, it has a real linear factor. A definite homogeneous quadratic vanishes only when both of its variables are zero, and its gradient also vanishes there. Even after embedding this zero set in $\mathbb R^n$, it has no real regular point. This contradicts the admissibility of $q$. Therefore at least one of the five conditions must hold.
\end{proof}

\subsection{Non-Pure-Real-Spectrum Case}
\label{subsec:complex-spectrum}

When there exist complex eigenvalues, the classification introduced in Section II.B is incomplete. This is because rotational phases can satisfy integer relations, and real forms of complex Jordan chains can generate invariants that are not products of squared moduli. 

Suppose that $q=0$ is a reduced rank hypersurface germ of $X_B$, whose degree is denoted by $d$. That is, $q\in\mathcal{H}_d$, and the dimension of $\mathcal{H}_d$
\begin{equation*}
    N_d:=\dim\mathcal H_d=\binom{n+d-1}{d}.
\end{equation*}
Since $\mathcal{H}_d$ is a vector space, any homogeneous polynomial $q\in\mathbb R[x_1,\ldots,x_n]$ can be expressed by a linear combination of the $N_d$ basis of $\mathcal H_d$. For instance, consider the case when $n=2$ and $d=2$. dim$\mathcal{H}_2=3$, where a basis is given by $x_1^2,x_1x_2$ and $x_2^2$. That is,
\begin{equation*}
    \mathcal{H}_2=\text{span}\{x_1^2, x_1x_2, x_2^2\}.
\end{equation*}
Every homogeneous polynomial of degree two with respect to $x_1$ and $x_2$ can be linearly expressed by the 3 basis.

Define the multi-index set
\begin{equation*}
    \mathcal A_{n,d}:=
    \{\alpha\in\mathbb Z_{\ge0}^n\mid |\alpha|=d\}.
\end{equation*}
Let $x^\alpha$ denote $\prod_{i}x_i^{\alpha_i}$. After fixing an ordering of the monomial basis $\{x^\alpha:\alpha\in\mathcal A_{n,d}\}$, let $H_d(B)\in\mathbb R^{N_d\times N_d}$ be the matrix representing the linear transformation on $\mathcal{H}_d$. According to \eqref{def lie derive}, $\mathcal{H}_d(B)$ can be constructed directly from
\begin{equation}
    \mathcal L_Bx^\alpha
    =\sum_{i,j=1}^{n}\alpha_iB_{ij}
       x^{\alpha-e_i+e_j},
    \label{induced-action}
\end{equation}
where $e_i$ denotes an $n$-dimensional vector with the $i$-th entry 1 and the others 0. The process of calculating $H_d(B)$ is illustrated by the following example. Suppose
\begin{equation*}
    B=\begin{bmatrix}
        b_{11}&b_{12}\\
        b_{21}&b_{22}\\
    \end{bmatrix}.
\end{equation*}
By \eqref{induced-action}, we can derive the following equations:
\begin{align*}
    \mathcal{L}_Bx_1^2&=2b_{11}x_1^2+2b_{12}x_1x_2,\\
    \mathcal{L}_Bx_1x_2&=b_{21}x_1^2+(b_{11}+b_{22})x_1x_2+b_{12}x_2^2,\\
    \mathcal{L}_Bx_2^2&=2b_{21}x_1x_2+2b_{22}x_2^2.
\end{align*}
Thus, under the set of monomial basis $\{x_1^2, x_1x_2, x_2^2\}$, we obtain
\begin{equation*}
    \mathcal{H}_2(B)=\begin{bmatrix}
        2b_{11}&b_{21}&0\\
        2b_{12}&b_{11}+b_{22}&2b_{21}\\
        0&b_{12}&2b_{22}\\
    \end{bmatrix}.
\end{equation*}

With the notation of $H_d(B)$, we can rewrite \eqref{darboux} as a matrix equation. Specifically, for a coefficient vector $a=(a_\alpha)_{\alpha\in\mathcal A_{n,d}}$, write
\begin{equation}
    q_a(x)=\sum_{\alpha\in\mathcal A_{n,d}}a_\alpha x^\alpha. \label{eq:q_a construct}
\end{equation}
Then
\begin{equation*}
    \mathcal L_Bq_a=\kappa q_a
\end{equation*}
if and only if
\begin{equation*}
    H_d(B)a=\kappa a.
\end{equation*}

\begin{theorem}
\label{thm:Darboux-criterion}
$B$ satisfies the RRHGC if and only if there exist
\begin{equation*}
    d\in\{2,\ldots,n\},\quad
    \kappa\in\mathbb R,
    \quad a\in\mathbb R^{N_d}\setminus\{0\}
\end{equation*}
such that:
\begin{enumerate}
    \item $H_d(B)a=\kappa a$;
    \item $q_a$ is irreducible in
          $\mathbb R[x_1,\ldots,x_n]$;
    \item there exists $x^*\in\mathbb R^n\setminus\{0\}$ satisfying
          \begin{equation*}
              q_a(x^*)=0,
              \quad \nabla q_a(x^*)\ne0.
          \end{equation*}
\end{enumerate}
\end{theorem}

\begin{proof}
If $B$ satisfies the RRHGC, take the irreducible homogeneous equation $q=0$ supplied by Definition~\ref{def:admissible-RRHG}. By Lemma~\ref{lem:invariance-identity}, $\mathcal L_Bq=\kappa q$ for a real constant $\kappa$. Expressing $q$ in the monomial basis of $\mathcal H_d$ yields $H_d(B)a=\kappa a$. The remaining two properties are consistent with Conditions 1) and 2) in Definition 1.

Conversely, the first condition and \eqref{eq:Darboux-flow} show that $q_a=0$ is invariant. Irreducibility, the degree bound, and the regular real zero show that it is an admissible non-hyperplane hypersurface germ. Thus $B$ satisfies the RRHGC.
\end{proof}

We now illustrate how to verify the three conditions in Theorem 2.

\textbf{Condition 1).} Fix an ordering of the monomial basis for $\mathcal{H}_d$. Then, we can construct $H_d(B)$ by \eqref{induced-action}. Compute every real eigenvalue $\kappa$ of $H_d(B)$ and parameterize the whole real eigenspace $E_{d,\kappa}=\ker(H_d(B)-\kappa I)$. $q_a$ can be derived as \eqref{eq:q_a construct}. We impose the normalization
\begin{equation*}
    \|a\|^2:=\sum_{\alpha\in\mathcal A_{n,d}}a_\alpha^2=1
\end{equation*}
to exclude the zero polynomial and remove the irrelevant multiplication of $q_a$ by a nonzero scalar.

\textbf{Condition 2).} Since $q_a$ is homogeneous, every nontrivial factorization of $q_a$ is also homogeneous.  Consequently, $q_a$ is reducible over $\mathbb R$ if and
only if there exists an integer
\begin{equation*}
    1\leq p\leq\left\lfloor\frac d2\right\rfloor
\end{equation*}
and nonzero homogeneous real polynomials of degrees $p$ and $d-p$ whose product equals $q_a$. It is sufficient to use $p\leq\lfloor d/2\rfloor$, since the two factors can be interchanged. For such a $p$, introduce coefficient vectors
\begin{equation*}
    b=(b_\beta)_{\beta\in\mathcal A_{n,p}},
    \quad
    c=(c_\gamma)_{\gamma\in\mathcal A_{n,d-p}},
\end{equation*}
and set
\begin{equation*}
    q_b(x)=\sum_{\beta\in\mathcal A_{n,p}}b_\beta x^\beta,
    \quad
    q_c(x)=\sum_{\gamma\in\mathcal A_{n,d-p}}c_\gamma x^\gamma.
\end{equation*}
The polynomial identity $q_a=q_bq_c$ is equivalent to the following equation
\begin{equation}
    a_\delta
    =\sum_{\substack{\beta\in\mathcal A_{n,p},\ 
                      \gamma\in\mathcal A_{n,d-p}\\
                      \beta+\gamma=\delta}}
       b_\beta c_\gamma, \quad \sum_{\beta\in\mathcal A_{n,p}}b_\beta^2=1, \label{factorization}
\end{equation}
where $\delta\in\mathcal A_{n,d}$. Normalizing $b$ does not restrict the possible factorizations. If $q_a=\widetilde q_b\widetilde q_c$ with $\widetilde q_b\neq0$, one may divide $\widetilde q_b$ by its coefficient norm and multiply $\widetilde q_c$ by the same norm. Hence the exact irreducibility condition is
\begin{equation*}
    \operatorname{Irr}_d(a)
    :=
    \bigwedge_{p=1}^{\lfloor d/2\rfloor}
    \neg R_{p,d-p}(a),
\end{equation*}
where $R_{p,d-p}(a)$ is \textit{True} when \eqref{factorization} holds.

\textbf{Condition 3).} The existence of a nonzero real regular zero can also be expressed by
polynomial equalities and inequalities.  Since $q_a$ is homogeneous, every nonzero zero can be rescaled to the unit sphere.  Therefore the required condition is that there exists $x^*\in\mathbb R^n$ such that
\begin{equation*}
\begin{aligned}
       q_a(x^*)=0,
       \quad \sum_{i=1}^{n}(x_i^*)^2=1,
       \quad \sum_{i=1}^{n}
       \left(
          \frac{\partial q_a}{\partial x_i}(x^*)
       \right)^2>0.
\end{aligned}
\end{equation*}
The last strict inequality is equivalent to $\nabla q_a(x^*)\neq0$.

Combining the preceding conditions gives the following formula:
\begin{equation}
\begin{aligned}
    &H_d(B)a=\kappa a, \quad \sum_{\alpha\in\mathcal A_{n,d}}a_\alpha^2=1,
    \\
    &\bigwedge_{p=1}^{\lfloor d/2\rfloor}
      \neg R_{p,d-p}(a),
    \\
    &q_a(x^*)=0,
      \quad\sum_{i=1}^{n}(x_i^*)^2=1,
    \\
    &\sum_{i=1}^{n}
      \left(
        \frac{\partial q_a}{\partial x_i}(x^*)
      \right)^2>0.
\end{aligned}
\label{final fomular}
\end{equation}

\begin{remark}
For matrices with non-pure-real eigenvalues, verifying the RRHGC essentially amounts to searching for a real homogeneous Darboux polynomial with a constant cofactor that is irreducible over $\mathbb{R}$ and whose zero set contains a nonzero real regular point. Although this search may become computationally demanding for high-dimensional systems, Theorem 2 provides a complete characterization and therefore determines whether the RRHGC is satisfied without overlooking any admissible reduced rank hypersurface germ.
\end{remark}

\section{Algorithm and Example}
\label{sec:complete-algorithm}

In this section, we provide an algorithm to verify the RRHGC based on Theorems 1 and 2. The algorithm first determines the spectral distribution of the matrix. If the spectrum is entirely real, it sequentially checks whether one of the five conditions in Theorem 1 is satisfied. Otherwise, it attempts to find a solution satisfying the three conditions in Theorem 2.

\begin{algorithm}[t]
\caption{Complete verification of the RRHGC}
\label{alg:complete-RRHGC}
\begin{algorithmic}[1]
\Require A matrix $B\in\mathbb R^{n\times n}$
\Ensure YES if and only if $B$ satisfies the RRHGC
\If{every eigenvalue of $B$ is real}
    \State Compute the Jordan data
    $\lambda_j$, $g_j$, $m_{j\ell}$ and $h$ in
    \eqref{eq:real-Jordan-form}.
    \If{$\max_{j,\ell}m_{j\ell}\ge3$}
        \State \Return YES \Comment{Condition \textnormal{(R1)}}
    \EndIf
    \If{$h\ge2$}
        \State \Return YES \Comment{Condition \textnormal{(R2)}}
    \EndIf
    \If{$g_j\ge3$ for some $j$}
        \State \Return YES \Comment{Condition \textnormal{(R3)}}
    \EndIf
    \If{$g_j\ge2$ and $g_k\ge2$ for some $j\ne k$}
        \State \Return YES \Comment{Condition \textnormal{(R4)}}
    \EndIf
    \ForAll{$c\in\mathbb Z^s\setminus\{0\}$ with $\gcd(|c_1|,\ldots,|c_s|)=1$ and         $|c_j|\le n$}
        \If{$\sum_jc_j=0$, $\sum_jc_j\lambda_j=0$, and
             $2\le\sum_jc_j^+=\sum_jc_j^-\le n$}
            \State \Return YES
            \Comment{Condition \textnormal{(R5)}}
        \EndIf
    \EndFor
    \State \Return \textsc{No}
\Else
    \For{$d=2,\ldots,n$}
        \State Construct $H_d(B)$ using \eqref{induced-action}.
        \If{there exist $\kappa\in\mathbb{R}$, $a\in\mathbb R^{N_d}$, $x^*\in\mathbb R^n$ such that \eqref{final fomular} holds}
            \State \Return YES \Comment{Theorem 2}
        \EndIf
    \EndFor
    \State \Return NO
\EndIf
\end{algorithmic}
\end{algorithm}

Example 1 is provided to illustrate the procedure of Algorithm 1.

\textbf{Example 1.} Consider the following matrix
\begin{equation*}
    B=
    \begin{bmatrix}
        1&-1&1&0\\
        1& 1&0&1\\
        0& 0&1&-1\\
        0& 0&1& 1
    \end{bmatrix}.
\end{equation*}
Its eigenvalues are $1+\mathrm i$ and $1-\mathrm i$, each with algebraic
multiplicity two.  The matrix contains a length-two complex Jordan chain, so
the Non-pure-real-spectrum branch of the algorithm applies.  In coordinates,
\begin{equation}
\begin{aligned}
    \dot x_1&=x_1-x_2+x_3,
    &\quad \dot x_2&=x_1+x_2+x_4,\\
    \dot x_3&=x_3-x_4,
    &\quad \dot x_4&=x_3+x_4.
\end{aligned}
\label{dynamics ex1}
\end{equation}

The procedure starts with $d=2$.  Since
\begin{equation*}
    N_2=\binom{4+2-1}{2}=10,
\end{equation*}
choose the ordered monomial basis
\begin{equation}
\begin{aligned}
    \mathcal M_2=\bigl(&x_1^2,x_1x_2,x_1x_3,x_1x_4,x_2^2,x_2x_3,x_2x_4,x_3^2,x_3x_4,x_4^2\bigr).
\end{aligned}
\label{basis}
\end{equation}
Direct application of the product rule to
\eqref{dynamics ex1} gives
\begin{align*}
\mathcal L_Bx_1^2
 &=2x_1^2-2x_1x_2+2x_1x_3,\\
\mathcal L_B(x_1x_2)
 &=x_1^2+2x_1x_2+x_1x_4-x_2^2+x_2x_3,\\
\mathcal L_B(x_1x_3)
 &=2x_1x_3-x_1x_4-x_2x_3+x_3^2,\\
\mathcal L_B(x_1x_4)
 &=x_1x_3+2x_1x_4-x_2x_4+x_3x_4,\\
\mathcal L_Bx_2^2
 &=2x_1x_2+2x_2^2+2x_2x_4,\\
\mathcal L_B(x_2x_3)
 &=x_1x_3+2x_2x_3-x_2x_4+x_3x_4,\\
\mathcal L_B(x_2x_4)
 &=x_1x_4+x_2x_3+2x_2x_4+x_4^2,\\
\mathcal L_Bx_3^2
 &=2x_3^2-2x_3x_4,\\
\mathcal L_B(x_3x_4)
 &=x_3^2+2x_3x_4-x_4^2,\\
\mathcal L_Bx_4^2
 &=2x_3x_4+2x_4^2.
\end{align*}
The induced matrix is
\begin{equation*}
H_2(B)=
\begin{bmatrix}
2& 1& 0& 0&0& 0&0& 0& 0&0\\
-2&2& 0& 0&2& 0&0& 0& 0&0\\
2& 0& 2& 1&0& 1&0& 0& 0&0\\
0& 1&-1& 2&0& 0&1& 0& 0&0\\
0&-1& 0& 0&2& 0&0& 0& 0&0\\
0& 1&-1& 0&0& 2&1& 0& 0&0\\
0& 0& 0&-1&2&-1&2& 0& 0&0\\
0& 0& 1& 0&0& 0&0& 2& 1&0\\
0& 0& 0& 1&0& 1&0&-2& 2&2\\
0& 0& 0& 0&0& 0&1& 0&-1&2
\end{bmatrix}.
\end{equation*}

Solving the real eigenvalue equation for $\kappa=2$ yields
\begin{equation*}
\begin{aligned}
E_{2,2}
=\ker(H_2(B)-2I)=\operatorname{span}\left\{
     e_8+e_{10},\ -e_4+e_6
   \right\}.
\end{aligned}
\end{equation*}
Therefore, every quadratic polynomial in this eigenspace has the form
\begin{equation*}
    q_{s,t}(x)
    =s(x_3^2+x_4^2)+t(x_2x_3-x_1x_4),
\end{equation*}
where $s,t\in\mathbb R$. Its coefficient norm in the basis \eqref{basis} is
\begin{equation*}
    \|a(s,t)\|^2=2s^2+2t^2.
\end{equation*}
Choose
\begin{equation*}
    s=0,
    \quad t=\frac1{\sqrt2}.
\end{equation*}
Then
\begin{equation*}
\begin{aligned}
    q_*(x)&=\frac{x_2x_3-x_1x_4}{\sqrt2},\\
    a_*&=\frac1{\sqrt2}
    (0,0,0,-1,0,1,0,0,0,0)^T,
\end{aligned}
\end{equation*}
and
\begin{equation*}
    \|a_*\|^2=1,
    \quad H_2(B)a_*=2a_*.
\end{equation*}

We next apply the exact irreducibility test.  Since $d=2$, only the split
$2=1+1$ must be considered.  Let
\begin{equation*}
\begin{aligned}
    q_b&=b_1x_1+b_2x_2+b_3x_3+b_4x_4,\\
    q_c&=c_1x_1+c_2x_2+c_3x_3+c_4x_4.
\end{aligned}
\end{equation*}
$R_{1,1}(a_*)$ is \textit{True} if $\sum_i b_i^2=1$ together with
\begin{align*}
b_1c_1&=0, &
b_2c_2&=0,\\
b_3c_3&=0, &
b_4c_4&=0,\\
b_1c_2+b_2c_1&=0, &
b_1c_3+b_3c_1&=0,\\
b_1c_4+b_4c_1&=-\frac1{\sqrt2}, &
b_2c_3+b_3c_2&=\frac1{\sqrt2},\\
b_2c_4+b_4c_2&=0, &
b_3c_4+b_4c_3&=0,
\end{align*}
which has no solution.

Finally, choose the unit vector
\begin{equation*}
    x^*=\left(\frac1{\sqrt2},0,
              \frac1{\sqrt2},0\right)^T.
\end{equation*}
Then
\begin{equation*}
    \|x^*\|^2=1,
    \quad q_*(x^*)=0,
\end{equation*}
and
\begin{equation*}
    \nabla q_*(x^*)
    =\left(0,\frac12,0,-\frac12\right)^T.
\end{equation*}
Consequently,
\begin{equation*}
    \sum_{i=1}^{4}
    \left(
      \frac{\partial q_*}{\partial x_i}(x^*)
    \right)^2
    =\frac12>0.
\end{equation*}
Thus, we have explicitly constructed a feasible reduced rank hypersurface germ with
\begin{equation*}
    d=2,
    \quad \kappa=2,
    \quad a=a_*,
    \quad x=x^*,
\end{equation*}
for formula \eqref{final fomular}. Therefore Algorithm 1 returns YES when $d=2$, and the degrees $d=3$ and $d=4$ need not be examined. The resulting reduced rank hypersurface germ is
\begin{equation}
    x_2x_3-x_1x_4=0.
    \label{hypersurface ex1}
\end{equation}
One can verify that
\begin{equation*}
    \mathcal L_B(x_2x_3-x_1x_4)
    =2(x_2x_3-x_1x_4),
\end{equation*}
so every trajectory starting on
\eqref{hypersurface ex1} remains on it.

\section{Conclusion}

This paper revisited the verification of the RRHGC, which is a technique in determining whether the reduced rank set of a strictly bilinear system contains reduced rank hypersurfaces. For matrices with only real spectrum, we established that the RRHGC holds if and only if at least one of five conditions is satisfied, thereby unifying and simplifying the previous real-spectrum criteria. For matrices with complex eigenvalues, we characterized the relevant invariant hypersurfaces using Darboux polynomials and reformulated RRHGC verification as the problem of finding a Darboux polynomial with a constant cofactor. Based on these results, a complete verification algorithm was proposed and illustrated through an example. The proposed algorithm can be used to address the special cases overlooked by the previous one and provides a more reliable basis for analyzing the global near-controllability of bilinear systems whose transition matrices form semigroups.

\bibliographystyle{IEEEtran}
\bibliography{main}

$\left. {}\right. $

\addtolength{\textheight}{-12cm}   


\end{document}